\documentclass[conference,page-numbering]{IEEEtran}

\usepackage{xy}
\usepackage{amsthm}
\usepackage{algorithm}
\usepackage{algpseudocode}
\usepackage{textcomp}
\usepackage{amsmath}
\usepackage{stmaryrd}
\usepackage{graphicx}
\usepackage{syntax}
\usepackage{fancyhdr}
\usepackage{textcomp}
\usepackage{flushend}

\theoremstyle{plain} 
\newtheorem{thm}{Theorem}[section]

\newtheorem{theorem}[thm]{Theorem}
\newtheorem{definition}[thm]{Definition}
\newtheorem{corollary}[thm]{Corollary}
\newtheorem{lemma}[thm]{Lemma}
\newtheorem{observation}[thm]{Observation}
\newtheorem{example}[thm]{Example}

\begin{document}
\xyoption{all}

\title{Abstract Predicate Entailment over Points-To Heaplets is Syntax Recognition}

\author{\IEEEauthorblockN{Ren\'{e} Haberland, Kirill Krinkin, Sergey Ivanovskiy$^{\dag}$}
  \IEEEauthorblockA{Saint Petersburg Electrotechnical University "LETI"\\
  Saint Petersburg, Russia\\
  haberland1@mail.ru, kirill.krinkin@fruct.org}
}

\maketitle
\thispagestyle{plain}
\pagestyle{plain}
\pagenumbering{gobble}

\begin{abstract}
Abstract predicates are considered in this paper as abstraction technique for heap-separated configurations, and as genuine Prolog predicates which are translated straight into a corresponding formal language grammar used as validation scheme for intermediate heap states. The approach presented is rule-based because the abstract predicates are rule-based, the parsing technique can be interpreted as an automated fold/unfold of the corresponding heap graph.
\end{abstract}

\IEEEpeerreviewmaketitle

\section{Introduction}

In order to understand the term of an \textit{abstract predicate} better let us consider a simple example from \textit{Computational Geometry} \cite{deBerg08}, the doubly-connected edge list. As the name tells us it is a list of edges, where each edge connects two two or three dimensional vertices, usually. If a mesh is triangulated each face consists primarily of three edges, however each vertex coincides twice with a neighbouring vertex, because one edge starts and one edge ends at a vertex. Each vertex may coincide with at least two other vertices. The face structure chosen might improve certain calculations, for instance, the normal vector or seeking for neighbouring faces, etc.

In a points-to heap a location points to some value, in case of the previous example that might be vertices or edges which are pointed to by some variable location. According to the mentioned doubly-linked edge list each list entry may point forward to the next edge or backward to the previous edge. In terms of a points-to heap it is only necessary to specify linked heap entries, non-linked/interleaved memory does not coincide by definition. Imagine, faces, which consist of at least three edges, would need to be specified every time for a heap in a points-to model. That is why a heap predicate may express a complex but intuitive situation very easily, for instance \texttt{face(p1,p2,p3)} may denote that three vertices \texttt{p1,p2,p3} are connected along a closed circuit building up a face, rather than specifying every time $\exists v1.v2.v3$, s.t. \texttt{p1.data $\mapsto$ $v1$} $\star$ \texttt{p2.data $\mapsto$ $v2$} $\star$ \texttt{p3.data $\mapsto$ $v3$} $\star$ \texttt{p1.next $\mapsto$ p2} $\star$ \texttt{p2.next $\mapsto$ p3} $\star$ \texttt{p3.next $\mapsto$ p1} $\star$ \texttt{p1.prev $\mapsto$ p3} $\star$ \texttt{p3.prev $\mapsto$ p2} $\star$ \texttt{p2.prev $\mapsto$ p1}.

First of all abstraction means a generalisation, often by introduction of an additional parameters. In this paper the parameter will be primarily Prolog terms, but the underlying assertion language could be abstracted too. By an \textit{abstract predicate} we would like to refer to an arbitrary predicate of the heap (here) in Horn clause form with any number of parameters, which may refer potentially to any number of further abstract predicates. Although introduced by some authors in the past with capital letters, we tend to characterise those special predicates only with an adjective ``\textit{abstract}'', we believe that should be fine. For the example above, \texttt{face(p1,p2,p3)} stands for the longer $\star$-separated heaplets. Depending on the actual proof entailment it might be more or less appropriate to perform a fold or an unfold of \texttt{face} causing a proof step to succeed, proceed, fail or block. However, in a fully automated proof, the decision may be hard to find. This paper formalises this problem and chooses a non-traditional approach at a first glance in order to resolve this issue for heaps.

Warren \cite{warren83} chooses the term ``\textit{Programming by Proving}'' to stress Prolog can be used as programming language, and what happens, in contrast to other programming languages, is Prolog is searching for a solution -- this is what Warren means by proving. Apart from that the \textit{Howard-Curry} isomorphism states there is an interconnection between proving on the one side and programming on the other side. W.r.t. heaplets, the main thesis behind this paper can be summarised in a simplistic way as ``\textit{Proving is Parsing}'', meaning by parsing one can prove correctness of specified heap behaviour and that heap behaviour is in a representation close to programming, namely to Prolog rules. A main observation is abstract predicates describe implicitly a language, and so a program's semantic problem can be resolved by pure syntax recognition.

The structure of this paper is as following: First, current approaches are presented, current restrictions and open questions are provided. Then, heaps and assertions are introduced formally. Abstract predicates and related defintions follow. Afterward the translation problem is formulated. Then the major contribution of this paper is presented -- the observation that abstract predicates entailment for points-to heaps can be reduced to syntax parsing, properties of the translation are considered. Finally, implementation details are discussed and conclusions are wrapped up.

\section{Current approaches}

Since the approach presented in this paper does seem to not follow traditional approaches, only a short review on the closest topics is thought to be meaningful.

First, abstract predicates as presented in this paper are supposed to be understood intentionally close to \cite{reynolds02,parkinsonThesis05}, the most important claim is there is a spatial operator $\star$ which conjuncts two separated heaps and furthermore, for which the rules and axioms of the \textit{Separation Logic} hold which is basically a substructural logic where the contraction rule does not necessarily hold. In this paper abstract predicates just serve as placeholders for $\star$-conjuncted heap expressions, so no fundamental changes to the core logic are made.

Abstract predicates may be controlled by the user, as it was implemented by verifast \cite{jacobs11}, or semi-automated as done with some tactics applicable to inductively defined predicates as part of Coq \cite{bertot04}, or fully-automated but with external \textit{fold-/unfold} \cite{Hutton98} hints for the automated proof entailment. Although the last approach is most convenient from a user's perspective it is  challenging from an algorithmic point of view.

Prolog is used in this paper as assertion language. \cite{kowalski74} demonstrates how Prolog can be used in order to express in general first-order predicate logic by Horn-clauses. \cite{warren83} presents what a Prolog semantics in terms of an abstract machine denotes to.
 Kallmeyer \cite{kallmeyer10} provides a good introduction into Prolog used for parsing. Particularly \textit{adjoint-trees} are proposed as recognition technique of natural language mutations which do not occur in formal or programming languages. However, Prolog is used intensively for implementation. \cite{Matthews98} demonstrates by example how Prolog can be used to address ambiguous parsing issues coming from natural language. Matthews discusses recursive tree recognisers in Prolog which seem to match LL(k)-parsers by obeying several modifications to a grammar which is provided as Prolog program: an accepting state needs to be made explicit, rule recursion is simulated by stack which is put onto parameter lists. Matthews uses \textit{difference lists} in order to implement in fact a \textit{partial derivation automaton} for regular languages (see \cite{Brzozowski64}). Finally he proposes Definite Clause Grammars and built-in commands in order to alter Prolog's knowledge base dynamically.
 
\cite{pereira012} can be considered as a de-facto standard reference on Prolog for natural language processing. From \cite{pereira012} one can find Prolog is incomplete because of possible left-recursive clauses, even when a solution to a given Prolog program exists in general. A model used in order to understand the structure of a natural sentence is a \textit{$\lambda$-annotated} parse tree which -- depending on its context -- may be interpreted arbitrarily.

\section{Points-to heaplets}
\label{PointsToHeapletsSection}

\begin{definition} A \textit{heap assertion} $H$ is inductively defined:
\begin{tabular}{ll}
H ::= & $\textbf{emp} \ | \ \textbf{true} \ | \ \textbf{false} \ | \ x \mapsto E \ | \ H \star H$\\
      & $ | \ H \wedge H \ | \ H \vee H \ | \ \neg H \ | \ \exists x.H \ | \ a(\vec{\alpha})$
\end{tabular}
\label{HeapAssertionDefinition}
\end{definition}

The assertion \textbf{emp} denotes an empty heap which by default is always true, it is the neutral element for ``$\star$'' which separates two non-intersecting heaps. The assertion \textbf{true} denotes any heap (which always is satisfied), where \textbf{false} denotes an arbitrary heap which always interprets as false. This definition is similar to Reynolds' definition \cite{reynolds02}. The core component of this definition is $x \mapsto E$, where $x$ is some location identifier (might be an object field, like $o1.field1$), and where $E$ is some valid assignable expression. This definition does no checks on types, this is what is supposed to be processed at an earlier stage \cite{haberland14}. In this paper it is also less of importance if it shall be the immediate meaning or just an address in heap space that stands on the right-hand side (compare with \cite{burstall72}).

Let us consider now, for example, two arbitrary heap predicate definitions formed into Prolog:
\begin{verbatim*}
p2(X,Y):-pointsto(loc2,X),pointsto(loc3,Y).
p1(X,Y):-pointsto(loc1,val1),p2(X,Y).
\end{verbatim*}

Here, \texttt{p2} denotes some predicate with two symbols \texttt{X} and \texttt{Y}, which are values pointed by fixed locations \texttt{loc2} and \texttt{loc3}. In contrast \texttt{p1} is different, hence it refers to predicate \texttt{p2}. Whenever we call \texttt{p2} with two syntactically valid term arguments we would have a the form $a(\vec{\alpha})$. Let us remind Prolog just does not find a solution for syntactically valid but semantically invalid terms (meaning the predicate's domain does not include such term).

Interpreting some heap formula $H$ for a given heap maps from a heap domain and a comparison heap into the boolean co-domain, meaning if the provided heap matches the existing heap the interpretation succeeds, and fails otherwise (applying only existing facts and rules which is equivalent to \textit{modus ponens}). The proof in Prolog succeeds iff a true goal is found -- a possibility to succeed or a refutation rejects the proof, which with no doubt is exactly the desired behaviour.

For sake of simplicity we agree to canonise heaplets, s.t. we conjunct them all along '$\star$', having the normal form $a_0 \star a_1 \star \cdots \star a_n$ or shorter $\prod^n_{\forall j} a_j$ for some $n \ge 0$. Further, $\wedge$ and $\vee$-connected heaps are turned into Prolog subgoal enumerations of kind $s_j, s_{j+1}, \cdots s_{j+k}$, alternatively split up into separate rule alternatives (or join using the ``;''-operator). The negation of an assertion is treated as an functor-guarded negate of a predicate, e.g. \texttt{not($P$)}, where \texttt{not} is a reserved keyword and $P$ some predicate. Just to be mentioned prior to the next sections -- negating a sequence of (sub-)goals means those may not follow whilst parsing. The quantification of a fresh variable is allowed spontaneously by simply introducing a fresh variable besides the existing predicate variables.

Moreover, we agree to keep in \textbf{emp}, \textbf{true} and \textbf{false}, although they are syntactic sugar. The heap assertion \textbf{emp} is sugar because it could be dropped entirely instead, \textbf{true} may be substituted by a tautology in its interpreting conjunction, \textbf{false} may be interpreted analogously.

\begin{corollary}
 Without going too deep into details, the previous definition implies any valid heap graph can be expressed by induction (and vice versa under the assumption indefinite elements are removed prior to transformation, such as \textbf{true} for instance). The proof(s) would be straight forward and inductively defined over the mentioned definition(s).
\end{corollary}

\begin{definition} A heap graph is a connected graph within the heap memory section and whose vertices may be pointed by at least one local variable from the stack memory. Each vertex is associated with a heap memory address, its length depends on its data structure. If a vertex is pointing to another vertex, both memory vertices coincide with a directed heap graph edge. If a vertex is pointed by two vertices, then one of which becomes an \textit{alias} of the other.
\end{definition}

\section{Abstract predicates}

Prolog is in this paper as programming language in which the assertion language for the heap is specified, and abstract predicates are part of it. Abstract predicates are defined as regular Prolog rules prior to using these rules later in an assertion formula. Often the assertion language does not match a formal specification, particularly for the heap, so this approach is an attempt to bridge this gap by using a logical programming language for logical reasoning, which abstract predicates are finally used for. For example, \cite{parkinsonThesis05} introduces an assertion language for predicates which is very independent, semantically and intentionally totally different from the surrounding programming and even assertion language, in fact (class-)typed variables are propagated typeless to predicates, program variables are simulated as being symbols rather than unidirectionally assignable locations, and the predicates are -- with big efforts -- tried to make up logical predicates as much as possible for what we  usually understand under a genuine (first-order) logical predicate.

The approach demonstrated in \cite{berdineCO05} introduces symbols for heaps, but not for denoting entire heaplets, so $X \star Y$ may not be used, for instance. In contrast to that, we allow symbols without such restrictions in Prolog, and we are not going to restrict ourselves to maximum matching rules only in general -- this does not imply there may not be introduced some tactics later on.

We use Prolog rather than an imperative programming language to specify the heap graph, because we believe the graph can be described better with predicates which are basically relations, and because the logical programming paradigm seems to be closer to the 1st-order predicate logic \cite{warren83} rather than the functional, for instance. It is also that the way heap assertions are supposed to be checked is closer to facts, rules and questions when it comes to reason logically about heap assertions.

Abstract predicates allow us to specify what the heap should look like, however the concern of compact specifications is due to the developer, regardless of how advanced abstract predicates are being processed.

The following formalisations will help us to describe the translation from abstract predicates into a grammar in the next section.

\begin{definition} A \textit{predicate rule} is defined as $\forall n.a:-q_{k \times n}.$ $\Leftrightarrow$ $a:-q_{k,0}, q_{k,1}, \ldots, q_{k,n}.$ for some arbitrary but fix integer $k$.
\end{definition}

It is said that $a$ holds whenever all its \textit{subgoals} $q_{k,j}$ hold for $0 \leq j \leq n$. The syntax of a predicate rule is defined as can be found in Fig. \ref{PrologRulesEBNF}.
\texttt{<number>} denotes any valid Prolog number, where \texttt{<atom> '(' <arguments> ')'} denotes some functor with atomic name \texttt{<atom>} and an arbitrary number of arguments. \texttt{<var>} denotes some variable symbol, which must start with an uppercase character letter, e.g. \texttt{X}. Fig. \ref{mapFunctionalExample} demonstrates an Prolog example.

\begin{figure}[t]
\begin{tabular}{c}
\begin{minipage}{7cm}
 \begin{grammar}
 <predicate> ::= <head> [ ':-' <body> ] '.'
 
 <head> ::= <atom> [ '(' <arguments> ')' ]
 
 <body> ::= <sub_goal> \{ ',' <sub_goal> \}*
 
 <sub_goal> ::= '!' | 'fail' | <functor_term>
  \alt <term> <rel> <term>
 
 <functor_term> ::= <atom> '(' [ <arguments> ] ')'
 
 <arguments> ::= <term> \{ ',' <term> \}*

 <term> ::= <atom> | <var> | <list> | <number>
   \alt <functor_term>
 
 <rel> ::= '=' | '\textbackslash =' | '\textless' | '\textless=' | '\textgreater' | '\textgreater='
 
 <list> ::= '[' [ <term> '|' ] <arguments> ']'
\end{grammar}
\end{minipage}
\end{tabular}
\caption{Extended Backus-Naur form for Prolog clauses}
\label{PrologRulesEBNF}
\end{figure}

Some predicate $a$ is evaluated by its subgoals left-to-right updating its symbol environment $\sigma$ every time:

\begin{tabular}{ll}
$C(a)\llbracket a(\vec{y}):-q(\vec{x}_{k,n})_{k \times n} \rrbracket \sigma =$ & \qquad\\
\multicolumn{2}{r}{\qquad $D \llbracket q_{k,n} \rrbracket \sigma(\vec{x}_{k,n}) \circ \cdots \circ D \llbracket q_{k,1} \rrbracket \sigma(\vec{x}_{k,1})$.}
\end{tabular}

By convention the term-vector $\vec{y}$ may intersect with $\vec{x}_{k,n}$, $\forall k,n$ and $C$ is of kind \texttt{atom} $\rightarrow$ \texttt{predicate} $\rightarrow \ \sigma \ \rightarrow \sigma$, $D$ has kind \texttt{subgoal} $\rightarrow \sigma \rightarrow \sigma$, and $\sigma$ is of kind $\texttt{term}^*$ $\rightarrow$ \texttt{term}, where $*$ denotes the Kleene-star operation for an arbitrary number of repetitions.

A subgoal $q_{k,j'}$ does not necessarily need to span a connected heap graph. However, if it does then obviously this does not only indicate some complete degree of \textit{separating concerns} which is a good pattern in software engineering, it also means that one abstract predicate actually pictures one entire problem locally. The corollary we can imply is: ``\textit{One abstract predicate shall correspond to one subheap}'', where a subheap contains a non-empty subset of vertices from the corresponding connected heap graph. Furthermore, by adding more and more $\star$-conjuncts we actually make the corresponding heap graph grow successively. The collection of $\star$-conjuncts forms a set of possibly connected with each other heaps which corresponds with abstract predicates, therefore abstract predicates in terms of points-to heaplets can be considered as a technique of specifying frames, or more generally speaking as a syntactic approach of specifying heaps. When talking about folding/unfolding of abstract predicates -- similar to function calls -- there exist parameters, namely heap graph vertices, which are available to both sides: a predicate's caller and callee side, and there are vertices that are only visible from within a predicate that cannot be references from the caller (at least not directly).

W.l.o.g. we agree that class objects field accessors, like \texttt{a.b}, are allowed according to Fig. \ref{PrologRulesEBNF} as \texttt{oa(object5,field123)} \cite{haberland14}. For sake of modularity and simplicity of demonstration and w.l.o.g., we further agree that class objects as well as object fields may be passed to predicates, and we do not need to worry about as long as the entire object is passed because in that case the treatment and behaviour does not change in comparison to regular automated variables as being mentioned later.

\begin{definition}
\label{PredicateRuleSetDefinition}
The \textit{predicate rule set} $\Gamma_a$ for some predicate name $a \in T$ and $\forall i,j. q_{i,j} \in (T\cup NT)$, where $T$ are terminals and $NT$ are non-terminals is defined as:\\\\
\begin{tabular}{ll}
$\Gamma_a ::=$ & $a:-q_{m \times n}$\\
 & $\Leftrightarrow
\begin{array}{llllllll}
 a:-    & q_{0,0} & , & q_{0,1} & , & \ldots & , & q_{0,m}\\
 \vdots & \vdots  &  & \vdots  &  & \ddots &  & \vdots\\
 a:-    & q_{m,0} & , & q_{m,1} & , & \ldots & , & q_{m,n}
\end{array}
$
\end{tabular}

If $m = 0$ then $a$ is a \textit{fact}. $a$ may be annotated by terms containing symbols (e.g. when $m=0$, $n>0$). If $t \in T$ then $t$ has the form $loc \mapsto val$, otherwise $t \in NT$ denotes the predicate name $t$ available in $\Gamma$.
\end{definition}

It is agreed that in a sequence $q_{k,0}, q_{k,1}, \ldots , q_{k,m}$ in $q_{m \times n}$ every line is canonised, such that for $s\leq m$ non-trivial entries the first $s$ subgoals are placed and that all remaining $m-s$ subgoals are tautologic subgoals with a domain entirely being true ($\top$).
Moreover, it is agreed that $\exists k.a:-q_{k} \preceq a:-q_{k+1}$ holds, meaning a predicate rule that occurs earlier in $\Gamma_a$ has a lower precedence than a predicate that is defined later.

\begin{corollary}
For the predicate environment $\Gamma$ of a Prolog program $\Gamma = \bigcup_{t \in T} \Gamma_t$ holds. All $\Gamma_t$ that depend on each other lay inside a \textit{predicate partition} $\overline{\Gamma_t}$. $\overline{\Gamma_t} \subseteq \Gamma$ holds.
\end{corollary}
\begin{proof}
 (sketch) The idea behind is to show all dependent $\forall t.\Gamma_t$ lay inside a partition, and all independent partitions, obviously, do not coincide with dependent predicate environments. Naturally, all predicate environments regardless if dependent or independent lay in $\overline{\Gamma}$. Predicates $\Gamma_a$ and $\Gamma_b$ from non-coinciding partitions in $\overline{\Gamma}$ can never depend on each other. 
\end{proof}
Remark: Obviously, due to the Halting problem the call of a predicate from a predicate partition may not terminate in general. Next, the expressibility of predicates is considered.

Remark: Possible naming clashes in $\Gamma$ may be resolved by mangling names including the location of the predicate, such as class where a predicate is defined etc. so the predicates become distinguishable. Predicates within the same location are believed to be together and hence do not clash by definition.
 
\begin{lemma}
\label{APsAreFOPsLemma}
Abstract predicates cover all first-order predicates.
\end{lemma}
\label{firstOrderLemma}
\begin{proof}
Please refer to \cite{kowalski74} for first-order predicate logic completeness expressibility of Prolog.
\end{proof}

\begin{lemma}
\label{APsAreSOPsLemma}
Abstract predicates may express second- (and even higher-) order logic predicates.
\end{lemma}

\begin{proof}
Up to this point we were only interested to know through Prolog we can express any first-order predicate logic formula. The following explanation shows we are not restricted ourselves to first-order, but we even can express higher-order in Prolog.
In Prolog the built-in predicate \texttt{call} allows to call a certain predicate with a list of input and output terms to be passed, for example \texttt{pred1(X):-call(pred2,X)}. 

For example, let us define the mapping of a predicate \texttt{P} on an input list, we agree the input parameters shall be encoded in \texttt{[X|Xs]} and output to be \texttt{[Y|Ys]}. Then the \texttt{map} predicate denotes as shown in Fig. \ref{mapFunctionalExample}. Higher-order predicates may be particularly of interest especially when it comes to dealing with class-objects and \textit{inversion of control}, as it is the case in many \textit{behavioural design patterns}, for instance in the Observer-pattern.

\begin{figure}[t]
\begin{tabular}{l}
\begin{minipage}{7cm}
\begin{verbatim}
map([],P,[]).
map([X|Xs],P,[Y|Ys]) :- 
    Goal =.. [P,X,Y],
    call(Goal), map(Xs,P,Ys).
\end{verbatim}
\end{minipage}
\end{tabular}
\caption{map/3 functional}
\label{mapFunctionalExample}
\end{figure}

The type of \texttt{map/3} is $list_a \rightarrow (list_a \rightarrow list_b) \rightarrow list_b$.
So, by introducing third and even higher-order predicates, in analogy to functions we may beat recursion in many cases by using an implicit recursion via higher-order predicates, for instance by application of a left fold that consumes some predicate $\oplus$ and applies it when appropriate having the following signature: \texttt{foldl($\oplus::a\rightarrow b \rightarrow a$, $\varepsilon$::$a$, $X$::$list_b$)::$a$} (right folding works in analogy to that). Foldl defines an algebra with an initial value $\varepsilon$ and carrier set $X$ and one operation $\oplus$ which is defined on the same type as $\varepsilon$ and element-wise for each element of $X$ and calculates a result of same type as $\varepsilon$. For example let us assume we have $a$ equal $b$ are integers and let $X=[1,2,3]$ be of kind ``\textit{list of integers}'', let us further agree our inital counter $\varepsilon$ equals 7, then \texttt{foldl} will calculate $((\varepsilon+1)+2)+3)$ which is 13 which is obviously an integer.
\end{proof}

For sake of completeness of the syntax definition from Fig. \ref{PrologRulesEBNF} and the translation in the following section we need to think about how to deal with ``\texttt{;}'' and ``\texttt{!}''. Actually, both are syntactic sugar. 

If the body of a predicate rule contains ``\texttt{;}'' then all right of it has to be split up into a fresh rule with the same name as the origin, so $b :- a_0, a_1, ..., a_m; a_{m+1}, ... , a_n$ is split up into $b :- a_0, a_1, ..., a_m.$ and $b :- a_{m+1}, ... , a_n.$. If there is a cut inside $b :- a_0, a_1, ..., a_m,!, a_{m+1}, ... , a_n$ then $a_0$, $a_1$ until $a_m$ may fail in which case other rules $b$ may be considered as alternatives if any existing. ``\texttt{!}'' makes sure that if only one subgoal only from $a_{m+1}$ until $a_n$ fails $b$ entirely fails without searching for alternatives. This is again, syntactic sugar, because all possible alternatives may be left-factorised so no other alternatives may be allowed -- so, this sugar would insist of rewriting existing predicate rule sets with the same name. This is why without any loss of generality ``\texttt{;}'' and ``\texttt{!}'' may in Prolog be be dropped from further considerations of expressibility.

Lemma \ref{APsAreFOPsLemma} and \ref{APsAreSOPsLemma} conclude that we are able to express all we would like in terms of Prolog, and that we could rewrite some predicate classes without explicit recursion. However, we do not intent to restrict ourselves in terms of $\mu$-recursive predicates.

\begin{definition}
\label{PredicateFoldingDefinition}
The \textit{predicate unfolding/folding} $a(\vec{\alpha})$ of/into some predicate $a$ for some rule system $\Gamma_a$ with actual term values $\vec{\alpha}$/subgoals $q_k$ is defined as:
(because of lemma \ref{APsAreSOPsLemma}) let $\Gamma_a$ be w.l.o.g. $a(\vec{y}):-q_k$ with $q_k = q_{k,0}(\vec{x}_{k,0}),q_{k,1}(\vec{x}_{k,1}), ... ,q_{k,m}(\vec{x}_{k,m})$. Now, if $\vec{\alpha} = (\alpha_0,\alpha_1,..,\alpha_A)$ and $\vec{y} = (y_0,y_1,..,y_A)$, then 
$a(\vec{\alpha}) \Leftrightarrow q_{k,0}(\vec{x}_{k,0}),q_{k,1}(\vec{x}_{k,1}), ... ,q_{k,m}(\vec{x}_{k,m})$
with $\alpha_0 \approx y_0, \alpha_1 \approx y_1, ... , \alpha_A \approx y_A.$
\end{definition}

In case of ``$\Rightarrow$'' of the above equivalence of $a(\vec{\alpha})$ a predicate is unfold. In case of ``$\Leftarrow$'' the right-hand side is folded into a predicate call. ``$\approx$'' stands for term unification.

\section{Interpreting abstract predicates over heap as syntax recognition}

The goal of syntax recognition is to automate the check of heap predicates in specifications against an inferred heap state. The technique applied is syntactic for a semantic problem. The problem with predicates is the non-determinism of when to fold/unfold. Proof tactics have been implemented in theorem provers and checkers, like Coq \cite{bertot04}, in very dedicated domains only, but the quality of the fold/unfold is still far from satisfiable. An automated fold/unfold approach would be highly desirable, so additional specifications or even manual interactions can be zeroed. The new approach presented in this paper requires the following steps:

\begin{enumerate}
 \item Translate incoming program and annotated assertions into Prolog terms which are integrated into \cite{haberland14}.
 \item Define abstract predicates within an annotated section in the incoming program. These Prolog rules need to be syntactically correct.
 \item Generate formal grammar for found abstract predicates. File grammar over to a parser-generator which will finally emit a valid and concrete parser.
 \item  While running the proof, trigger certain parse rules depending on abstract predicate calls found.
\end{enumerate}

It will provide a different non-standard way of dealing with the problem.

\begin{observation}
\label{FormalLanguageObservation}
When looking at how a heap is generated it reminds a production system for formal languages.
\end{observation}

Terminals become points-to expressions (cmp. with definition \ref{PredicateRuleSetDefinition}) or relations, and non-terminals become abstract predicate subgoals. Terminals are concatenated, where points-to expressions are loosely coupled with $\star$. $\star$ is commutative. Points-to expressions may be concatenated too, when the pairs are ordered according to the left-hand side location name. If local names clash, a namespace would resolve this by full qualification, e.g. by name prefixes.

\begin{theorem}
The predicate partition builds up an production system and manifests in fact a context-free grammar.
\end{theorem}

\begin{proof} The left-hand side of a predicate rule may only be no more than one non-terminal. It is more than regular because there is no such claim the right-hand side needs to be right-recursive. If the head of the rule contains arguments this still does not change statically the dependencies in between the predicates. A predicate partition has one starting non-terminal.
\end{proof}

\begin{observation}
When looking at how a heap is derived from abstract predicates one may think about reducing it.
\end{observation}

The implication underneath, however, would be both, inferred heap state and expected heap specification, still contain some folded predicate definitions which shall be unfolded until both sides establish an equivalence. This would be a bi-directional approach. However, that problem could be reduced to \textit{Post's Correspondence Problem} and is unfortunately undecidable in general, hence is not considered here any further.\\\\
Observation \ref{FormalLanguageObservation} seems promising, so the heap predicate check can be re-formulated as: ``\textit{Given a $\star$-connected heap, does it match a given heap specification or not?}''. But there might be further questions related, such as: ``\textit{What would be the closest correct heap, s.t. it satisfies the current heap specification?}'', which could deliver us answer to the counter-example problem.

\begin{lemma}
The word-problem for abstract predicates $P$ can be formulated as: Given $\alpha_1, \alpha_2 \in L(G(P))$ does $\alpha_1 \equiv \alpha_2$ hold ? $G(P)$ denotes the formal context-free grammar obtained from the predicate partition of $P$.
\end{lemma}

\begin{proof}
Here $\alpha = (a+A)^*$, and $a \in T$, $A \in NT$. $T$ denotes all terminals which are parameterised and encode source and target of ``$\mapsto$'', $NT$ denotes non-terminals which contain all predicates and parameterise all valid input terms. The rule set $P$ is the translated set of predicates, $L(G(P))$ is the language generated by the generated grammar by Prolog rules. The starting non-terminal is a predicate call from either $\alpha_1$ or $\alpha_2$. Regardless of the particular kind of parser to be used the follow set $\sigma(\alpha)$ and the first-terminal sets $\pi(\alpha)$ need to be calculated (see later). One important fact is there is not a single start non-terminal, but there might be several depending on number of predicate calls in $\alpha_1$ and $\alpha_2$. Furthermore, there is not only one path searched from $\alpha_1 \vdash^* \alpha_2$ but also $\alpha_2 \vdash^* \alpha_1$. Only if there is no path found in both directions $\alpha_1$ does not coincide with $\alpha_2$, otherwise it does. In order to check if two sentences match, it is not only necessary to construct paths between predicates, it is also necessary to consume initial and intermittent \texttt{pointsto}-terminals. Parameters on terminals and non-terminals shall be bound according to the current binding and unified with $\alpha_1$ and $\alpha_2$.
\end{proof}

\section{Translation of Horn-clauses into grammar}

This section considers how abstract predicates provided as Prolog rules are translated into a general context-free grammar.
Before that, Prolog rules need to be analysed lexically, so all expressions of form $loc \mapsto val$ are considered tokens. Multi-paradigmatic programming \cite{denti05} allows interpreting Prolog rules during execution of some main Prolog application, which, in our case, would be the verification. This process only needs to be done once until all abstract predicates have been processed. The translation process from Prolog rules into a formal grammar is astonishing simple. However, Prolog rules may have argument terms on the left and the right side of ``\texttt{:-}'', this can be modelled by introducing attributes to the generated grammar, we finally obtain an attributed grammar. Hence, the translation  $C\llbracket \rrbracket$ can be defined straight:

\begin{definition}
$C\llbracket \rrbracket$ is defined as rule transducer for an incoming abstract predicate set and attributed grammar as output:\\

\begin{tabular}{l}
 $C\llbracket \rrbracket = \emptyset$\\
 $C\llbracket C_1 . C_2 \rrbracket = C\llbracket C_1 \rrbracket \ \dot{\cup} \ C\llbracket C_2 \rrbracket$\\
 $C\llbracket a(\vec{x}):-q^{0}(\vec{x}), ..., q^{n}{(\vec{x})} \rrbracket$ = $\{ a_{\vec{x}} \rightarrow q^0_{\vec{x}} ... q^n_{\vec{x}}  \}$
\end{tabular}
\end{definition}

In contrast to previous notations subgoals here have upper indices and $\vec{x}$ now accommodates all variable symbols within of a predicate rule for notational comfort, so if a particular subgoal $q^j$ for some $j$ does not require all components of $\vec{x}$ then it does not. Remind $\dot{\cup}$ is a set union where the element sequence matters. As can be seen from both notations are pretty similar to each other and are interchangeable, the inverse operation $C^{-1}\llbracket \rrbracket$ translates an attributed grammar back into Prolog and can be defined as:

\begin{definition}
$C^{-1}\llbracket \rrbracket$ is defined as rule transducer for an incoming attributed grammar and an abstract predicate set as output:\\

\begin{tabular}{l}
 $C^{-1}\llbracket \rrbracket = \emptyset$\\
 $C^{-1}\llbracket C_1 \ C_2 \rrbracket = C^{-1}\llbracket C_1 \rrbracket \ . \ C^{-1}\llbracket C_2 \rrbracket$\\
 $C^{-1}\llbracket a_{\vec{x}} \rightarrow q^0_{\vec{x}} ... q^n_{\vec{x}} \rrbracket$ = $\{ a(\vec{x}):-q_0(\vec{x}), ..., q_n(\vec{x})\}$
\end{tabular}
\end{definition}

\begin{corollary}
$C \llbracket \rrbracket$ and $C^{-1} \llbracket \rrbracket$ terminate for any well-defined domain input.
\end{corollary}

\begin{proof}
 The proof is rather trivial, since there is no infinite cycle possible. Both, $C \llbracket \rrbracket$ and $C^{-1} \llbracket \rrbracket$ linearly scan all incoming rules successively from left to right. Suppose, there was a cycle in between particular rules. Even so, both translations will finally terminate because cycles may bother only while parsing, not while translating. The starting point to a predicate partition corresponds one to one to the starting non-terminal of a subgrammar. There might be several entry points for abstract predicates, and so are the entry points corresponding to non-terminals for a grammar.
\end{proof}

We still need to investigate $C \llbracket \rrbracket$ and $C^{-1} \llbracket \rrbracket$ soundness and completeness.

\begin{corollary}
$C \llbracket \rrbracket$ and $C^{-1} \llbracket \rrbracket$ are total and both mappings are complete and sound.
\end{corollary}

\begin{proof}
 It is not hard to verify that $C \circ C^{-1} \circ C \equiv C$ hold as well as $C^{-1} \circ C \circ C^{-1} \equiv C^{-1}$ by simple substitution of the definitions from above. Because of the discussions in section \ref{PointsToHeapletsSection}, ``\texttt{!}'' nor ``\texttt{;}'' do not matter w.r.t. expressibility. If, however, the domain of $C\llbracket \rrbracket$ is supposed to not terminate, then its co-domain will cause exactly the same behaviour, same holds for $C^{-1}\llbracket \rrbracket$.
\end{proof}

\section{Parsing}

For the purpose of a simple and intuitive algorithm the constants from definition \ref{HeapAssertionDefinition} will not be considered. Because as mentioned earlier they are not intrinsic and can be dropped therefore. Essentially those heap constants provide partial heap expressiveness and may be considered for future research, w.r.t. class objects a \textbf{true} could possibly mean, for instance, to consume all \texttt{pointsto} until a (rule-dependent) marker pops up indicating a \textit{safe} synchronisation point in terms of error productions \cite{Grune90} for the ongoing parsing as described briefly. The input word is finite, however in general the number of unfolds may be hypothetically infinitely many -- but shown later the $\pi$-function allows us to pre-calculate the following terminal symbols with polynomial efforts.

This section introduces fundamental conventions necessary to complete some generic LL(k)-parser for demonstration purposes. Fundamentally, this is not only for an unlimited forward-looking LL parser, it is also possible to use some different parser, for instance based on a generalised LR or Earley parser. First, a sentence is defined as some composite of \texttt{pointsto} (terminals) and further subgoals with term arguments (non-terminals) -- something that the right-hand side of a (Prolog) rule comes up with. Second and third, in analogy to a LL(k)-parser but with functional space rather than single character as for strings the definitions of \textit{first} and \textit{follow} sets are introduced. Forth, both SHIFT and REDUCE are proposed for some general parser implementation.

\begin{definition}
An \textit{abstract sentence} $\alpha$ is a $\star$-conjunction of heaps which are denoted by $a \mapsto b$, where $a$ is some unique location identifier and $b$ some value object or \textbf{nil}.
\end{definition}

For example, $\alpha::=$ \texttt{[ pointsto(x,nil),} \texttt{pointsto(y,1), member(x,[y])]} may describe the current state of the heap during the verification of an imperative program, and $\star$ is replaced in the previous list by commas. The specification of a rule may insist on \texttt{[pointsto(Y,1),member(X,[Y|_]),pointsto(X,_)]}. So, what is necessary to check the abstract sentence from the program matches the sentence from the specification is primarily to check whether all parts from either of both is derivable from each other.

An abstract sentence may also contain term unification, such as \texttt{pointsto(X,5),X=Y}. Term unification has to be thoroughly analysed and separated from the remaining two cases, namely \texttt{pointsto} which denotes terminals, and predicate calls as subgoals which denote non-terminals later on. It has to be taken into consideration that not any terms may be unified, since there is a occurs-check taking place by intention w.r.t. the given implementation and by default in Prolog, so indefinite terms or recursive term definitions are strictly prohibitted.

Let us now formulate an initial algorithm in order to check two abstract sentences match or do not match, let us consider algorithms 1. $\pi$ denotes the function being introduced shortly. The problem is we reduce (factual predicate unfolding) possibly ad absurdum, we do not know determined when and how often to fold and unfold which obviously also depends on the rules themselves. Assume, we had some $\alpha_1 = [\underbrace{a \mapsto b}_{i_0},i_1,\cdots, i_{m_1}, \underbrace{q_1(x)}_{p_0}]$ and some $\alpha_2 = [\cdots \underbrace{a \mapsto b}_{j_3},\cdots, \underbrace{q_1(x)}_{q_7}]$ we would like to match against with. SHIFT-TERMs would first of all unify $i_0$ with $j_3$ and possibly continue with all other matching terms. REDUCE-PREDs will try match all matching predicates which may have to be unfolded first, that is why the first terminal of a predicate may be required first, the expansion $expand(p_k,\alpha_1)$ expands the predicate head by the corresponding body of $p_k$ (compare with definition \ref{PredicateFoldingDefinition} and Fig. \ref{PrologRulesEBNF}) into a new abstract sentence $\alpha'$ which might be described in Prolog as \texttt{concat($\alpha$,$[i_7,i_8,i_9],\alpha'$)}, if for instance $q_1(x)$ unfolds into $[i_7,i_8,i_9]$.

\begin{algorithm}[t]
\caption{A na\"{\i}ve algorithm for word equality check for abstract sentences; \textbf{Input:} $\alpha_1$ = $[i_0,...,i_{m_1},$ $p_0,...,p_{n_1}]$, $\alpha_2$ = $[j_0,...,j_{m_2},q_0,...,q_{n_2}]$ with $i$ and $j$ as \texttt{pointsto}-terminals, and subgoals $p$ and $q$ as non-terminals. $\Gamma$ contains all pre\-dicate definitions. \textbf{Result:} \textit{True} in case $\alpha_1$ equals $\alpha_2$, \textit{false} otherwise.}
\begin{algorithmic}[1]
\Procedure {SHIFT-TERMs}{$\Gamma$, $\alpha_1$, $\alpha_2$}
 \ForAll {$k \in \{0, ... , m_1\}$}
   \If {$\exists l. l \in \{0, ..., m_2\} \wedge i_k \approx j_l$}
     \State $\alpha_1 \leftarrow \alpha_1 \setminus i_k$
     \State $\alpha_2 \leftarrow \alpha_2 \setminus j_l$
   \EndIf
  \EndFor
\EndProcedure
\Statex
\Procedure {REDUCE-PREDs}{$\Gamma$, $\alpha_1$, $\alpha_2$}
 \ForAll {$k \in [0 .. n_1]$}
   \Comment{try match with terminal}
   \If {$\exists l. l \in [0..m_2] \wedge j_l \in \pi(p_k)$}
     \State $expand(p_k, \alpha_1)$
   \Else
   \Comment{try to reduce in $\alpha_2$}
     \If {$\exists l. l \in [0..n_2] \wedge (\pi(q_l) \cap \pi(p_k) \ne \emptyset) $}
       \State $\alpha_1 \leftarrow expand(p_k, \alpha_1)$
       \State $\alpha_2 \leftarrow expand(q_l, \alpha_2)$
     \Else
       \Comment{Match}
       \If {$m_1 = m_2 = n_1 = n2 = 0$}
         \State $true \rightarrow Halt!$
       \Else
         \Comment{Non-Match}
         \State $false \rightarrow Halt!$
       \EndIf
     \EndIf
   \EndIf
  \EndFor
\EndProcedure
\end{algorithmic}
 \label{AlgorithmEqualityCheck}
\end{algorithm}

\begin{definition}
The \textit{first set} is defined as co-domain of a total map $\pi$ with type $(T \cup NT) \rightarrow 2^T$ for $m > 0$, such that:\\

$\pi(a) ::= 
\left\{
	\begin{array}{ll}
		a  & \mbox{if } a \mbox{ is } X \mapsto Y \mbox{ or } \Gamma_a ::= a.\\\\
		\bigcup_{0 \leq j \leq n} \pi(q_{j,0}) & \mbox{if } \Gamma_a ::= a:- q_{m \times n}, n > 0.
	\end{array}
\right.
$
\end{definition}

Essentially, $\pi$ determines all terminals that start with \texttt{pointsto} or are beginnings (only the first terminal) of predicate subgoals (independent of its arguments).

\begin{definition}
The \textit{follow set} $\sigma(t) \subseteq T$ for $t \in (T \cup NT)$ is defined as:\\

$\sigma(t) ::= 
\left\{
	\begin{array}{lll}
		\bigcup_{i,j} \pi(q_{i,j+1})  & \mbox{if } t \mbox{ is at } (i,j<n) \mbox{ in } q_{m \times n}\\
		                              &            \wedge \ 0 \leq i \leq m\\
		                              &            \wedge \ q_{i,j+1} \neq \top\\
		                              &            \wedge \ \exists a.\Gamma_a ::= a:-q_{m \times n}\\\\
		\bigcup_{a} \sigma(a)                & \mbox{if }  t \mbox{ is at } (i,n) \mbox{ in } q_{m \times n}\\
		                              &            \wedge \ \Gamma_a ::= a:-q_{m \times n}\\
		                              &            \wedge \ \exists b.\Gamma_b ::= b:-q_{m_b \times n_b}\\
		                              &            \wedge \ a \mbox{ is at } (i_b,j_b) \mbox{ in } q_{m_b \times n_b}\\\\
		\emptyset                     & \mbox{otherwise}
	\end{array}
\right.
$
\end{definition}

The follow set determines literally all terminals that may follow a \texttt{pointsto} or a predicate subgoal from all considered rules. Now we have defined $\pi$ and $\sigma$ we are able to synthesise from this a LL(k) recogniser (\cite{Grune90} might be a helpful introduction).

\begin{example}
Given the following production rules $ q_1 \rightarrow a, \  q_2 \rightarrow a q_2 \ | \ q_3 b, \ q_3 \rightarrow \varepsilon \ | \ q_3 a$ these rules are obviously ambiguous, for instance in $\pi(q_2)=\{a\}, \sigma(a)=\{ \varepsilon\} \cup \pi(q_2) \cup \pi(q_3) \cup \sigma(q_3)$.
\end{example}

\begin{example}
Given the following finite specification $[ (loc1,v1), p1(loc1,loc2), (loc2,v2) ]$ and the word $[ (loc1,v1), (loc2,v2)  ]$ which is true, but only if $p1$ does not dump a heap.
\end{example}

In case the input word is not particularly a sequence of terminals, but an abstract sentence, it will be necessary to cut redundant calculations as early as possible. Hence, memoizing calculated subgoals would not only enhance the speed of search (since only the predicate name and its parameters play a role in memoization), it would resolve matching the first non-terminal issue, which by the way, matches neatly in the described definition of $\pi$ and $\sigma$.

Negated predicates are dropped here, refer to section \ref{Implementation} for details.

\section{Properties}

In Fig. \ref{heapConfiguration} a sample heap configuration is shown. This configuration consists of 8 triangles, where each triangle is crossed by either a dotted or a tortuous line. The line from the midpoint of $v_0$ and $v_3$ to $M_1$ denotes the triangle $\Delta(v_0,v_3,M_1)$, where the tortuous line between the midpoint of $v_0$ and $v_1$ and $M_1$ addresses $\Delta(v_0,M_1,v_1)$. And so, Fig. \ref{heapConfiguration} demonstrates there might be more than one way of spanning the heap graph by some provided heap predicate set, namely either by the triangles marked with dotted lines or with the tortuous lines. It is, however, essential that all vertices need to be included in a heap in order to decide heap graph isomorphism.\\

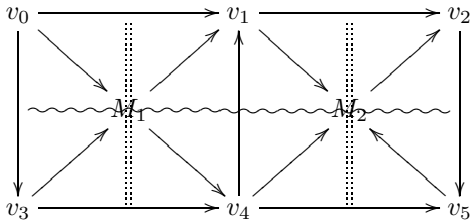
\begin{figure}[b]
\begin{center}
\begin{tabular}{l}
\begin{minipage}{7cm} 
\xymatrix{
 v_0 \ar[rr] \ar[rd] \ar[dd] &  \ar@{:}[dd]        & v_1 \ar[rr]        \ar[rd] & \ar@{:}[dd] & v_2 \ar[dd]\\
        \ar@{~}[rrrr]        & M_1 \ar[ru] \ar[rd] &                            & M_2 \ar[ru] & \\
 v_3 \ar[ur] \ar[rr]         &                     & v_4 \ar[uu]\ar[ru] \ar[rr] &             & v_5 \ar[lu]                         
}
\end{minipage}
\end{tabular}
\caption{A sample heap configuration where $v_j$ with $j \ne 2, j \ne 5$ outgoing edges are of class object type with more than one associated pointer attributes}
\label{heapConfiguration}
\end{center}
\end{figure}

There is (currently) only one position where a non-deterministic decision has to be taken out, namely the decision where to bound the input stream w.r.t. object boundaries. If introducing a convention that only common fields of an object are allowed and are canonised, the decision becomes determined. Obeying the convention could be performed within polynomial effort, as the rest of the parsing routine -- which is well-described and is known to be tractable within polynomial efforts (refer to \cite{Grune90} on parsing foundations). If all predicate partitions are parsable, as described with possible practical restrictions discussed in the previous sections, then entailment over points-to heaplets becomes decidable and finally terminates with an answer or proof refutation. The proof refutation will be in fact be a syntax error with corresponding coordinates in the points-to encoded input word according to the expected predicate partition.

The core memory model has not been modified nor extended, except the introduction of abstract predicate definitions over the existing points-to model. The proposed extension shall therefore by compatible e.g. with Reynolds' \cite{reynolds02} or Burstall's model \cite{burstall72}, however in this notation in contrast to Burstall cell addresses where not used, so this approach is conventionally closer to Reynolds. The consequences of Burstall's notation could be researched, since it is common practice compound objects are referred in practice by reference, not by its content.

\section{Implementation}
\label{Implementation}

The implementation is in GNU Prolog and uses ANTLR version 4 \cite{Parr12}, and is supposed to incorporate the framework presented in \cite{haberland14}, \cite{haberland15}. Initially both tools were chosen for simplicity and extensibility reasons and for lecturing purposes.

In order to gain from flexibility and a huge support of existing software packages, the chose programming paradigm is 
\textit{multi-paradigmatic} \cite{denti05}, which allows the developer to write and run programs in different programming languages and profit from both of its advantages. The integration of both works astonishing simple due to a Proxy design pattern and an interpretation in both directions. There exists also an experimental user interface based on tuProlog for prototyped development.
The implementation first translates input language into an intermediate representation which are Prolog terms, afterwards the assertions are copied separately into a Prolog theory, and abstract predicates are transformed into a ANTLR 4 grammar file as explained earlier. Whenever a parsing is requested, an internal syntax recognition process is initialised in the language the ANTLR outputted the recognisers (which is Java here). The output and control is passed back to the invoker. This way abstract predicates can be checked fully automated, and in case of an error the corresponding error will be processed.

ANTLR makes use of so-called ``\textit{syntactic and semantic predicates}'' in order to fight syntax ambiguity. Since ANTLR does not necessarily cover in practice all LL(k)-grammars strictly, there is of course still room for improvement. Practically this means that occasionally there may appear grammars which shall, but which do not parse due to current limitations of the ANTLR parser generator. There were made experiences other parser generators, e.g. LR(0)-parsers resolve this issue and even recognise left-recursion by definition, but lack from known shift-reduce restrictions on the other side therefore, like with GNU bison, for instance. A good introduction to parsing techniques can be found in \cite{Grune90}.

As an example of required transformation during the analysis of lexemes and tokens some precautions were required. First, $bar \mapsto foo$ is mangled to $pt\_3bar\_3foo$ where 3 is the length of the name or corresponding accessor. If the accessor is compound, e.g. $b.f.g$, then the accessor length avoid ambiguity. For instance \texttt{pointsto(X,2)} is mangled to a Prolog atom $p\_X\_2$. If needed, a mangled name can be demangled -- the internal parsing may be done by a Java helper which is then visible in Prolog as a helper built-in predicate via \cite{denti05}.
Second, the left-hand side (de-)canonisation (on Prolog level).
\texttt{p1(X,[X|Y]):-...} is transformed into \texttt{p1(X1,X2):-X1=X,X2=[X|Y],...}.
Third, a Prolog rule \texttt{p(X,Y):-$\alpha$} may be translated to a ANTLR-grammar snippet: 
\texttt{p[String X,String Y]:}
$\alpha$. This way all synthesised attributes may be passed top-down, inherited attributes may be modified to some predicate \texttt{p} by adding \texttt{returns} together with the inherited attribute names just before colon. So, all what is necessary is to decide whether a variable is inherited or synthesised in order to decide its position in the grammar snippet.

When abstract predicates are turned into a concrete grammar, e.g. a ANTLR grammar file, the problem arises that unified terms are together with \texttt{pointsto} terminals and non-terminals. Unifying terms need to be separated from terminals and non-terminals, therefore they are moved into translating rules beside the attributed grammar. For instance, ANTLR introduces translating rules using the curly brackets within a sentence.
Negated sentences and fragments of it can be introduced by ``$\sim$'' and brackets, and mean the included sentence may not appear. No further cases are discussed since either by attributes or translating rules additional behaviour may be mimiced, such as a failing predicate as a parser signal.

\section{Conclusions}
The approach presented proposes a technique to automatically entail points-to heaplets by syntax analysis rather than manually fold/unfold abstract predicates. If a predicate partition is representable as valid set of a parser's rule set then there will be definite answer whether heap specification and a existing heap configuration match. The model used in between Prolog and a concrete parser rule set is an attributed grammar \cite{Grune90} which translates inheriting and synthesising attributes which correspond to Prolog's head terms. We believe the implementation of stack frames which is different to those of common imperative programming language, gives Prolog an important advantage in reasoning since it is what we would expect from the attribute content control -- without the need of additional implementation nor development costs, because it is part of Prolog's core \cite{warren83}.  The used points-to heaplets may correspond to a Separation Logic styled model obeying its axioms and rules. It is true, Prolog's abstract machine is an interpreter and therefore on average slower than any natively compiled code, but the question of performance is minor in this case since we are mostly interested in exploring tractability and expressibility first of all - since verification is done separately from the generated code, we take intentionally the risk of being a bit slower occasionally.

Instead of \textit{simulating} only symbols within an artificially introduced assertion language, the assertions here are all expressed in the same language in which proofs will be taken out, Prolog. There is no overhead of mapping in between assertion and proof language as it is for instance the case in \cite{bertot04}. Symbols may be used very closely to the first-order predicate logic, symbols and terms may be unified, which is a bi-directional reasoning technique. It is believed expressibility in Prolog terms, especially with regards to objects, may improve over non-adequate representations, as it was proven concept in \cite{haberland08} on markup-notations for terms. It remains an open question due to term unification and the rule-based predicate partitions if and how error production rules may in fact advance the completeness of reasoning rules further, for instance w.r.t. abduction and proof explanation. Another practical advantage using Prolog whilst proving is the possibility to load parts or none of Prolog rules and facts in order to try some question without loading all at once, which makes error tracing comfortable.

Open questions and future work includes the possibility of partial heap specifications using constant keywords like \textbf{emp}, \textbf{true} and \textbf{false}, and the question if a proof may get simpler just if the heap graph is required to be connected. The last question arises for improved error location on proof refutations using error production rules to be introduced and invoked during parsing. It could stop, for instance, on a refutation and behave like in a ``\textit{panic mode}'' \cite{Grune90} consuming all input tokens until a synchronised save state or a sequence of consecutive known safe tokens.

\section*{Acknowledgement}
\dag This article is dedicated to Sergey Ivanovskiy in remembrance.


Parts of this paper are prepared as a contribution to the
state project of the Board of the Ministry of Education of the Russian Federation (task \# 2.136.2014/K).



\renewcommand{\refname}{References}

\end{document}